\documentclass[10pt]{sig-alternate-05-2015}
\usepackage{authblk}

\usepackage[T1]{fontenc}
\usepackage[hidelinks]{hyperref}

\usepackage{amsfonts}
\usepackage{amssymb}

\usepackage{mathtools}
\usepackage{cite}
\usepackage[linesnumbered,algoruled,noend]{algorithm2e}
\makeatletter
\newcommand{\removelatexerror}{\let\@latex@error\@gobble}
\makeatother
\let\oldnl\nl
\newcommand{\nonl}{\renewcommand{\nl}{\let\nl\oldnl}}
\usepackage{paralist}
\usepackage{graphicx}

\usepackage{tikz}                                                               
\usepackage{subcaption}

\usetikzlibrary{calc,shapes,arrows,chains,backgrounds,decorations.pathmorphing}
\usepackage{color}

\definecolor{myred}{HTML}{FC8D59}
\definecolor{myyellow}{HTML}{FFFFBF}
\definecolor{mygreen}{HTML}{99D594}

\usepackage{etoolbox}
\usepackage{comment}

\newtoggle{APPENDIX}
\togglefalse{APPENDIX}
\newtoggle{TR}
\togglefalse{TR}
\toggletrue{TR}
\makeatletter
\def\@copyrightspace{\relax}
\makeatother

\newtheorem{definition}{Definition}
\newtheorem{theorem}{Theorem}
\newtheorem{lemma}{Lemma}
\newtheorem{corollary}{Corollary}

\DeclarePairedDelimiter\abs{\lvert}{\rvert}
\makeatletter
\let\oldabs\abs
\def\abs{\@ifstar{\oldabs}{\oldabs*}}
\makeatother

\newcommand\set[1]{\left\lbrace #1 \right\rbrace}

\newcommand{\equalityref}[1]{\hyperref[#1]{Equality~\eqref{#1}}}
\newcommand{\inequalityref}[1]{\hyperref[#1]{Inequality~\eqref{#1}}}

\newcommand{\MA}{message adversary}

\newcommand{\MAD}{\mathtt{MA}}

\newcommand{\EStable}{\lozenge  \mathtt{STABLE}(D)}

\newcommand{\ra}{\rightarrow}
\newcommand{\G}{\mathcal{G}}

\newcommand{\Gr}{\mathcal{G}^r}

\newcommand{\Seq}[3]{(\G^{#1})_{#1 = #2}^{#3}}
\newcommand{\Seqr}[2]{\Seq{r}{#1}{#2}}

\newcommand{\SeqrI}{(\G^r)_{r \in I}}

\newcommand{\Exec}[1]{\langle #1 \rangle}

\SetKwFunction{roots}{roots}
\newcommand{\nbound}{N}
\newcommand\Star[1]{\mathcal{S}({#1})}
\newcommand\Starwl[1]{\mathcal{S}'({#1})}
\newcommand\SeqG[1]{(#1)_{D+1}^{\infty}}
\newcommand\SeqGr[1]{(#1)_{r+1}^{\infty}}

\DeclareMathOperator{\CPast}{CP}
\DeclareMathOperator{\IN}{In}

\renewcommand{\geq}{\geqslant}
\renewcommand{\leq}{\leqslant}
\renewcommand{\epsilon}{\varepsilon}
\newcommand{\CP}[3]{\CPast_{#1}^{#3}({#2})}

\newcommand\incG{\ensuremath{\overline{\G}}}
\newcommand\incR{\ensuremath{\overline{R}}}
\newcommand{\todo}[1]{}
\renewcommand{\todo}[1]{{\color{red} TODO: {#1}}}
\newcommand\indist[1]{\sim_{#1}}
\newcommand\Alg{\mathcal{A}}
\newcommand\MASty[1]{\mbox{\footnotesize \sf{\textup{#1}}}}
\newcommand{\DBMA}{\MASty{DIAM}(D)}
\newcommand{\BStable}{\lozenge \MASty{STABLE}_D}
\newcommand{\GStable}{\lozenge  \MASty{STABLE}^*}
\newcommand{\BLiveness}{\lozenge \MASty{STABILITY}}
\newcommand{\BSafety}{\MASty{ROOTED}}
\newcommand{\OldSafety}{\MASty{STICKY}}
\renewcommand{\MAD}{\MASty{MA}}
\renewcommand{\EStable}{\BStable}

\newcommand{\Ra}{\Rightarrow}

\newcommand{\postr}{t}
\newcommand{\VSRC}{stable root component}
\newcommand{\VSRCaps}{Stable Root Component}
\newcommand{\Rooted}[1]{$#1$-rooted}

\SetKwData{Proposal}{prop}
\SetKwData{Lock}{prop}
\SetKwData{Locked}{locked}
\SetKwData{Queue}{queue}
\SetKwFunction{Root}{root}
\newcommand\true{\textsc{true}}
\newcommand\false{\textsc{false}}

\newcommand\lastround{\ensuremath{r^*}}
\newcommand\InvA{\ensuremath{\mathcal{I}}}
\newcommand\InvB{\ensuremath{\mathcal{I}'}}
\newcommand\JnvA{\ensuremath{\mathcal{J}}}

\newcommand{\impscale}{0.9}

\newcommand{\namedref}[2]{\hyperref[#2]{#1~\ref*{#2}}}
\newcommand{\sectionref}[1]{\namedref{Section}{#1}}

\newcommand{\theoremref}[1]{\namedref{Theorem}{#1}}
\newcommand{\defref}[1]{\namedref{Definition}{#1}}
\newcommand{\figureref}[1]{\namedref{Figure}{#1}}
\newcommand{\lemmaref}[1]{\namedref{Lemma}{#1}}

\newcommand{\algref}[1]{\namedref{Algorithm}{#1}}
\newcommand{\corollaryref}[1]{\namedref{Corollary}{#1}}

\newcommand{\lineref}[1]{\namedref{Line}{#1}}
\newcommand{\equref}[1]{\hyperref[#1]{Eq.~\eqref{#1}}} 

\newcommand{\SOURCE}{\operatorname{SOURCE}}
\newcommand{\QUORUM}{\operatorname{QUORUM}}

\providecommand{\keywords}[1]{\textbf{\textit{Keywwords: }} #1}

\begin{document}

\title{\iftoggle{TR}{}{Regular Submission: }Consensus in Rooted Dynamic Networks with Short-Lived Stability}


%


\author[1]{Kyrill Winkler}
\author[1]{Manfred Schwarz}
\author[1]{Ulrich Schmid}
\affil[1]{TU Wien, Vienna, Austria\\ \{kwinkler, mschwarz, s\}@ecs.tuwien.ac.at}

\maketitle
\begin{abstract}
	We consider the problem of solving consensus using deterministic algorithms in a
	synchronous dynamic network with unreliable, directional point-to-point
	links, which are under the control of a message adversary.
        In contrast to a large body of existing work that focuses on oblivious message adversaries
        where the communication graphs are picked from a predefined set,
        we consider message adversaries where guarantees about stable periods that occur
        only eventually can be expressed.
	We reveal to what extent such eventual stability is necessary and sufficient, that is,
        we present the shortest period of stability that permits solving consensus,
	a result that should prove quite useful in systems that exhibit
	erratic boot-up phases or recover after repeatedly occurring, massive transient faults.
	Contrary to the case of longer stability periods, where we show how standard algorithmic
	techniques for solving consensus can be employed, the short-lived nature of the stability
        phase forces us to use more unusual algorithmic methods that avoid waiting explicitly for
        the stability period to occur.
\end{abstract}

%

\iftoggle{TR}{}{\clearpage}
\section{Introduction}
\label{sec:intro}

We consider deterministic consensus algorithms in synchronous dynamic networks,
where a potentially unknown number
$n$ of processes that never fail\footnote{Nevertheless, a crash of
	some process $p$ in some round could easily be modelled by 
$p$ sending no messages in any later round.} 
communicate via unacknowledged messages over
unreliable point-to-point links. 
Consensus, which is a pivotal service in truly distributed applications, is
the problem of computing a common decision value based on local input values
of all the processes.  
An execution of a consensus algorithm in our system proceeds in a sequence of
lock-step synchronous\footnote{It assumes that all processes simultaneously
broadcast a message at the beginning of a round, receive the messages from 
each other, and then simultaneously make a state transition at the end 
of the round, thereby proceeding to the next round.}
rounds, where message loss is modelled using an omniscient message adversary that determines the
directed \emph{communication graph} $\G^r$ for each round $r$. A directed edge $(p,q)$
present in $\G^r$ means that the message sent by $p$ in round $r$ is successfully
received by $q$.

In most existing work in this area, e.g.\ \cite{SW89,SWK09,CGP15,AG13}, the
message adversary is oblivious, i.e., may choose each $\Gr$ from the
\emph{same} set of admissible graphs arbitrarily in each round.
For instance, the classic result from Santoro and Widmayer \cite{SW89} states
that consensus is impossible if the adversary may suppress $n-1$ or more
messages in every round.
More recently, \cite{CGP15} introduced an equivalence relation on the set of admissible communication graphs such that
consensus is solvable if and only if for each equivalence class there is
a common \emph{source} (a node that has a directed path to every 
other node) in every graph. These (and similar) approaches characterize the solvability 
of consensus by means of properties of the \emph{set} of admissible graphs.

We also explore the solvability/impossibility border of
consensus, albeit under \emph{non-oblivious} message adversaries that
support \emph{eventual stabilization}
\cite{BRS12:sirocco,SWSBR15:NETYS,SWS16:ICDCN}: Here, the set of admissible
choices for $\Gr$ may change with evolving round numbers~$r$. Rather than
constraining the set of admissible graphs, we hence constrain admissible
graph \emph{sequences}. As it
turns out, consensus can be solved for graph sequences where
the \emph{set} of graphs occurring in the sequence would render consensus 
impossible under an oblivious message adversary \cite{SW89,CGP15}.

Apart from being theoretically interesting, considering eventually
stabilizing dynamic networks is also useful from a practical 
perspective: Algorithms that work correctly under eventually stabilizing message
adversaries are particularly suitable for systems that suffer from uncoordinated boot-up
sequences or systems that must recover from massive transient faults: Network connectivity can be expected to improve over time here, e.g., due
to improving clock synchronization quality. Since it is usually
difficult to determine the time when such a system has reached
normal operation mode, algorithms that just terminate when
a reasonably stable period has been reached are obviously advantageous.
Algorithms that work correctly under short-lived stable periods are particularly
interesting, since they have higher coverage and terminate earlier in
systems where longer stable periods occur only rarely or even not at 
all. Note that the occurrence of short-lived stability periods could be
confirmed in the case of a prototype wireless sensor network \cite{PS16:SENSORCOMM}.

Last but not least, stabilizing algorithms require less reliable and, in our
case, not inherently bidirectional communication underneath, hence 
work with cheaper and/or more energy-efficient network communication
interfaces.
After all, guaranteeing reliable bidirectional communication links typically incurs
significant costs and/or delays and might even be impossible in adverse
environments. We hence conjecture that our findings may turn out useful for
applications such as mobile ad-hoc networks \cite{KM07} with heavy
interference or disaster-relief applications \cite{LHSP11}.

In view of such applications, our core assumption of a synchronous
system may appear somewhat unreasonable. However, it is not thanks 
to modern communication technology \cite{SCS04}: As synchronized clocks 
are typically required for basic communication 
in wireless systems anyway, e.g., for transmission scheduling and 
sender/receiver synchronization, global synchrony 
is reasonably easy to achieve: It can be integrated
directly at low system levels as in 802.11 MAC+PHY 
\cite{IEEE802Phy}, provided by GPS receivers, 
or implemented by means of network time synchronization
protocols like IEEE~1588 or FTSP \cite{MKSL04}.

\textbf{Main contributions and paper organization:} In this paper, we thoroughly answer
the question of the minimal stability required for solving consensus 
under eventual stabilizing message adversaries. After the introduction
of our system model and our message adversaries in \sectionref{sec:model} and 
\sectionref{sec:adversary}, respectively, we establish the following results:
\begin{enumerate}
\item[(1)] We provide a novel algorithm in \sectionref{sec:algorithm}, along with its correctness proof, which
solves consensus for a message adversary that generates graph sequences
consisting of graphs that (i) are rooted, i.e., have exactly one root component (a strongly connected 
component without any incoming edges from outside of the component),
and (ii) contain a subsequence of $x=D+1$ consecutive graphs 
whose root component is formed by the same set of nodes (``stable root component'').
Herein, the system parameter $D \leq n-1$ is the dynamic diameter, 
i.e., the number of rounds required for any node in a stable root component to reach
all nodes in the system. 
Thanks to (i), our algorithm is always safe
in the sense that agreement is never violated; (ii) is only needed to ensure
termination. Compared to all existing algorithms for non-oblivious 
message adversaries like \cite{BRS12:sirocco,SWSBR15:NETYS,SWS16:ICDCN}, where the
processes more or less wait for the stability window to occur, our algorithm uses 
quite different algorithmic techniques.
\item[(2)] In previous work \cite{BRS12:sirocco,BRSSW16:TR}, it has been shown that
$x=D-1$ is a lower bound for the stability interval for any consensus 
algorithm working under message adversaries that guarantee a stable root component to occur eventually, and that (a bound on) $D$
must be a priori known.\footnote{Whereas this may seem a somewhat unrealistic 
(though inevitable) restriction at first sight, it must be noted that 
$D$ only needs to be guaranteed throughout the stability interval.
And indeed, our wireless sensor network measurements \cite{PS16:SENSORCOMM} 
confirmed that this is not an unrealistic assumption.} 
In \sectionref{sec:impossibility} of this paper, 
we improve the lower bound to $x=D$, 
which reveals that the previous bound
was not tight and that our new algorithm is optimal.
This result also shows that the mere propagation of some input 
value to every process
does not suffice to solve consensus in this setting. 
\item[(3)] To complement earlier results \cite{SWS16:ICDCN}
about consensus
algorithms that work for stability periods longer than $2D+1$, 
we show in \sectionref{sec:completepicture}
that very large
periods of stability, namely, at least $3n-3$, 
also allow to adopt the well-known uniform
voting algorithm \cite{CBS09} for solving consensus in our setting.

\end{enumerate}
Some conclusions and directions of future work in \sectionref{sec:conclusion} complete the paper.

\medskip

As a final remark, we note that our methodology
is in stark contrast to the approach advocated in \cite{RS13:PODC}, 
which shows, among other insightful results, that the message adversary 
$\SOURCE+\QUORUM$ allows to simulate an asynchronous message passing system
with process crashes augmented by the failure detector $(\Sigma,\Omega)$.
Since this is a weakest failure detector for consensus \cite{DFGHKT04},
it is possible to use classic consensus algorithms on top of this
simulation.
Furthermore, as $\Sigma$ is the weakest failure detector to simulate shared
memory on top of wait-free asynchronous message passing \cite{DFGHKT04}, even
shared memory algorithms that rely on $\Omega$ could be employed.

In \cite[Sec.~8]{BRSSW16:TR}, we hence investigated the 
potential of simulating $(\Sigma,\Omega)$ on top of eventually
stabilizing message adversaries, as this would allow us to employ such 
well-established consensus solutions instead of specifically
tailored algorithms.
Unfortunately, it turned out that $\Sigma$ cannot be
implemented here, even in the case of message adversaries 
that eventually guarantee
an infinite period of stability --- not to speak of message adversaries
that guarantee only a finite period of stability like the one presented in
this paper.
Therefore, we had to conclude that, for this type of message adversaries,
failure detector simulations are no viable alternative to the approach 
taken here.

\subsection*{Related work}

Research on consensus in synchronous message passing systems 
subject to link failures dates back at least to the seminal
paper \cite{SW89} by Santoro and Widmayer; generalizations 
have been provided in 
\cite{SWK09,CBS09,BSW11:hyb,CGP15,CFN15:ICALP}. In all these papers,
consensus, resp.\ variants thereof, are solved in systems where,
in each round, a digraph is picked from a set of possible 
communication graphs. The term message adversary was coined 
by Afek and Gafni in \cite{AG13} for this abstraction.

A different approach for modeling dynamic networks has been
proposed in \cite{KLO10:STOC}: $T$-interval connectivity
guarantees a common subgraph in the communication graphs
of every $T$ consecutive rounds. \cite{KOM11} studies agreement 
problems in this setting. Note that solving consensus is
relatively easy here, since the model assumes bidirectional and
always connected communication graphs.
In particular, $1$-interval-connectivity, the weakest form of
$T$-interval connectivity, corresponds to all nodes constituting a perpetually constant set of
source nodes.

In both lines of research, there is no notion of eventually
stabilizing behavior of dynamic networks. To the best of
our knowledge, the first instance
of a message adversary that guarantees eventual stable root components has been considered in 
\cite{BRS12:sirocco}: It assumed communication graphs
with a non-empty set of sources and long-living periods of stability $x=4D+1$.
\cite{SWSBR13:PODC,SWSBR15:NETYS} studies consensus
under a message adversary with comparably long-lived stability,
which gracefully degrades to general $k$-set agreement in case
of unfavorable conditions. However, this message adversary 
must also guarantee a certain influence relation between 
subsequently existing partitions. \cite{SWS16:ICDCN}
established a characterization of uniform consensus solvability/impossibility
for longer stability periods. In particular, it provides a consensus
algorithm that works for stability periods of at least $2D+1$ but
does not require graph sequences where all graphs are rooted.

Finally, \cite{RS13:PODC} used message adversaries that allow a notion of ``eventually
forever'' to establish a relation to failure detectors. Albeit we
do not consider this ``extremal'' case in this paper, which solely
addresses short-lived stability, we note that interesting
insights can be drawn from this relation.

\section{Model} \label{sec:model}

We consider a set $\Pi$ of $n$ deterministic state machines,
called \emph{processes}, which communicate via message passing over unreliable
point-to-point links. Processes have unique identifiers and are typically 
denoted by $p, q$.
We call algorithms that do not depend on $n$ \emph{uniform} algorithms,
whereas algorithms that rely on at least some bound on $n$ are called \emph{non-uniform}.
In this paper, we will consider exclusively non-uniform algorithms, except for
\theoremref{thm:uniform-impossibility}, where we touch upon uniform algorithms as well.
 Processes never fail and operate
synchronously in lock-step rounds $r=1,2,\dots$, where each round 
consists of a phase of communication followed by a local state
transition of every process. In the communication phase of a round,
every process sends a message (possibly empty) to every other process 
in the system, and records the messages successfully received from the 
other processes. A \emph{message adversary} (see \sectionref{sec:adversary} for
the detailed definitions)
is a set of graph sequences that determine which messages are lost in each round.

The \emph{state} of a process $p$ at the end of its round $r$ computation
is denoted by $p^r$, and the collection of the round $r$ states of all processes
is called round $r$ \emph{configuration} $C^r$.
Those messages that are delivered by the message adversary in a given round $r>0$
are specified
via a digraph\footnote{We sometimes write $p \in \Gr$ instead of $p \in \Pi$ to stress that $p$ is a vertex of $\Gr$,
and sloppily write $(p \rightarrow q) \in \Gr$ instead of $(p \rightarrow q) \in E(\Gr)$.}
$\Gr = \langle \Pi, E^r \rangle$,
called the round~$r$ \emph{communication graph}.
An edge $(p \rightarrow q)$ is in $\Gr$ if and only if the round $r$ message of
$p$ sent to $q$ is not lost. We assume that every process $p$ always successfully
receives from itself, so the self-loops $(p \rightarrow p)$ are in every $\Gr$.
The \emph{in-neighborhood} of $p$ in $\Gr$, $\IN_p(\Gr) = \{ q \mid (q, p) \in \Gr)$
hence represents the processes from which $p$ received a message in round $r$.

A \emph{message adversary} is characterized by the set of infinite sequences of consecutive communication graphs
that it may generate, which are called \emph{admissible}. A sequence 
$\sigma$ of consecutive communication graphs, ranging from round $a$ to round $b$, is denoted as $\sigma = \Seqr{a}{b}$, where
$|\sigma| = b-a+1$, with $b=\infty$ for infinite sequences. Since we actually identify 
a message adversary with its set of admissible sequences, we can  
compare different message adversaries via a simple set inclusion.

We consider the \emph{consensus problem}, where each process $p$ starts with
some input value $x_p$ and has a dedicated write-once output variable $y_p$; eventually, every process needs to irrevocably
decide, i.e., assign a value to $y_p$
(\emph{termination}) that is the same at every process (\emph{agreement})
and was the input of some process (\emph{validity}).
The assignment of the input values for each process is summarized in
some \emph{initial configuration} $C^0$.
Given a message adversary $\MAD$ and a deterministic consensus algorithm $\Alg$, 
an (admissible) \emph{execution} or \emph{run} 
$\varepsilon = \Exec{C^0, \sigma}$ is uniquely determined by $C^0$ and
an admissible graph sequence $\sigma\in \MAD$.

Applying $\Alg$ and a finite sequence $\sigma'$ to a configuration $C$ of $\Alg$
yields the configuration $C' = \Exec{C, \sigma'}$ of $\Alg$.
As usual, we write $\varepsilon \indist{p} \varepsilon'$ if the finite or
infinite executions
$\varepsilon$ and $\varepsilon'$ are \emph{indistinguishable} to $p$
(i.e., the round $r$ state of $p$ is the same in both executions)
until $p$ decides.

\subsection*{Dynamic graph concepts}

First, we introduce the pivotal notion
of a \emph{root component} $R$, often called root for brevity,
which denotes the vertex-set of a
strongly connected component of a graph
where there is no edge from a process
outside of $R$ to a process in $R$.
\defref{def:root} gives its formal definition.

\begin{definition}[Root Component]\label{def:root}
  $R \neq \emptyset$ is a root
  (component) of graph $\G$, if it is the set of vertices of a strongly connected 
  component $\cal R$ of $\G$ and
  $\forall p \in \G, q \in R : (p \ra q) \in \G \Rightarrow p \in R$.
\end{definition}

It is easy to see that every graph has at least one root component.
A graph $\G$ that has a \emph{single} root component is called 
\emph{rooted}; its root component is denoted by $\Root(\G)$.
Clearly, a graph $\G$ is rooted if and only if it has a
rooted spanning tree: the root component is the union of the
roots of all the spanning trees of $\G$.
Hence, there is a directed path from every node of $\Root(\G)$ to every other node
of $\G$.

Conceptually, root components have already been employed
for solving consensus a long time ago:
The asynchronous consensus algorithm for initially dead processes introduced
in the classic paper \cite{FLP85} relies on a suitably constructed initial 
clique, which is just a special case of a root component.

In order to model stability, we rely on root components that
are present in every member of a (sub)sequence of communication graphs.
We call such a root component the \emph{\VSRC} of a sequence and stress that,
albeit the set of processes remains the same, the interconnection topology
between the processes of the root component and to the processes outside
may vary greatly from round to round.


\begin{definition}[\VSRCaps{}]\label{def:common-root} \label{def:single-root}
  We say that a non-empty sequence $\SeqrI$ of graphs has a \VSRC{} $R$,
  if and only if each $\Gr$ of the sequence
  is rooted and 
  $\forall i, j \in I : \Root(\G^i) = \Root(\G^j) =  R$.
  We call such a sequence a \Rooted{R} sequence.
\end{definition}

We would like to clarify that while ``rooted'' describes a graph property,
``\Rooted{R}'' describes a property of a sequence of graphs.



Given two graphs $\G = \langle V, E \rangle$,
$\G' = \langle V, E' \rangle$ with the same vertex-set $V$, let the
\emph{compound graph} $\G \circ \G' := \langle V, E'' \rangle$ where
$(p, q) \in E''$ if and only if for some $p' \in V:$ $(p, p') \in E$ and $(p', q) \in E'$.
Since we assume self-loops the compound graph can be written as 
the product of the adjacency matrices of $\G$ and $\G'$.

In order to model information propagation in the network, we
use a notion of \emph{causal past}:
Intuitively, a process $q$ is in $p$'s causal past, denoted
$q \in \CP{p}{r'}{r}$ if $q=p$ or if, by round $r$, $p$ received information
(either directly or via intermediate messages) that $q$ sent in round
$r'+1$.


\begin{definition}[Causal past]\label{def:causal-influence}
Given a sequence $\sigma$ of communication graphs that contains rounds $a$ and $b$,
the causal past of process $p$ from (the end of) round $b$ down to
(the end of) round $a$ is $\CP{p}{a}{b} = \{ p \}$ if $a = b$ or 
$\CP{p}{a}{b} = \IN_p(\G^{a+1} \circ \cdots \circ \G^b)$ if $a < b$.
\end{definition}

A useful fact about the causal past is that in
full-information protocols, where processes exchange their entire
state history in every round, we have $q \in \CP{p}{s}{r}$ if and
only if, in round $r$, $p$ knows $q^s$, the round $s$ state of $q$.

To familiarize the reader with our notation, we conclude this section with the following
technical \lemmaref{lem:information-propagation}.
It describes the information propagation in a graph sequence containing an
ordered set $G = \set{\G^{r_1}, \ldots, \G^{r_n}}$,
$i \neq j \Ra r_i \neq r_j$, and $i > j \Ra r_i > r_j$,
of $n$ distinct communication graphs, where any
$\G, \G' \in G$ are both rooted, but
$\Root(\G)$ is not necessarily the same as $\Root(\G')$.
As we have mentioned earlier, every $\G \in G$ has 
hence a rooted spanning tree and is therefore weakly connected.
In essence, the lemma shows that, by the end of round $r_n$,
each process $p$ received a
message from some process $q$ that was sent after $q$
was member of a root component of some graph of $G$.

\begin{lemma}\label{lem:information-propagation}
  Let $G = \set{\G^{r_1}, \ldots, \G^{r_n}}$ be an ordered set of rooted
  communication graphs 
  and let $X \subseteq \Pi$ with
  $X \cap \Root(\Gr) \neq \emptyset$ for every $\Gr \in G$.
  For every $p \in \Pi$, there is some $q \in X$ ($q$ may depend on
$p$) and a $\G^r \in G$ s.t.\
  $q \in \Root(\G^r)$ and $q \in \CP{p}{r}{r_n}$.
\end{lemma}

\begin{proof}
  Let $S^{r_i}$ be the set of those processes that, in round $r_i$, have received information from a 
  process of $X$ after it was member of a root component so
  far in a graph in $G$.
  Formally,
  $S^{r_i} := \{p \in \Pi \mid
    \exists q \in X,
    \exists r \in \{r_1, \ldots, r_i \}  \colon
    q \in \Root(\Gr) \land q \in \CP{p}{r}{r_i}\}$.
  In order to show the lemma, we prove by induction on $i$ from $1$ to $n$ that
  $\lvert S^{r_i} \rvert \geq i$.
  We use the abbreviation $X^{r_i} := X \cap \Root(\G^{r_i})$.
  The base of the induction, $\lvert S^{r_1} \rvert \geq 1$, follows from the observation that
  $X^{r_1} \subseteq S^{r_1}$.
  For the induction step, assume for $1 \leq i < n$ that $\lvert S^{r_i} \rvert \geq i$.
  We show that then $\lvert S^{r_{i+1}} \rvert \geq i+1$.
  If $\lvert S^{r_i} \rvert \geq n$, as obviously 
  $S^{r_i} \subseteq S^{r_{i+1}}$, we are done.
  If $\lvert S^{r_i} \rvert < n$, consider that
  $X^{r_{i+1}} \subseteq S^{r_{i+1}}$.
  If $\lvert X^{r_{i+1}} \setminus S^{r_i} \rvert \geq 1$, we immediately have
  $\lvert S^{r_{i+1}}\rvert > \lvert S^{r_i} \rvert$.
  If $\lvert X^{r_{i+1}} \setminus S^{r_i} \rvert = 0$,
  note that there is a path in $\G^{r_{i+1}}$ from every $q \in X^{r_{i+1}}$ to every
  $p \in \Pi$.
  Hence, $(u \ra v) \in \G^{r_{i+1}}$ for some $u \in S^{r_{i}}$,
  $v \in \Pi \setminus S^{r_{i}}$, since we assumed $\abs{S^{r_i}} < n$.
  As $u \in S^{r_i}$, there is some $q \in \bigcup_{r=r_1}^{r_i} X^r$ with
  $q \in \CP{u}{r}{r_i}$.
  By \defref{def:causal-influence}, since $(u \ra v) \in \G^{r_{i+1}}$, we have
  $q \in \CP{v}{r}{r_{i+1}}$ and thus $v \in S^{r_{i+1}} \setminus S^{r_i}$ which
  implies
  $\lvert S^{r_{i+1}} \rvert > \lvert S^{r_i} \rvert$.
\end{proof}

\tikzset{
chain/.style={decorate, decoration={snake, segment length=8pt,
pre length=0.5mm, post length=2mm}, ->},
lossy/.style={<->, densely dotted, >=stealth}}
\tikzset{nchain/.style={decorate, decoration=zigzag, ->}}

\newcommand\head{\mbox{head}}
\newcommand\tail{\mbox{tail}}
\newcommand\rTERM{\tau}
\newcommand\BBCommGraph{\useasboundingbox (-.7, -.9) rectangle (3.4, 1.3);}
\newcommand\BBSingleGraph{\useasboundingbox (-.7, -.9) rectangle (2.3, 1.3);}
\newcommand\BBExecLabel{\useasboundingbox (-.2, -.7) rectangle (.2, 1.1);}

\section{Message Adversaries} \label{sec:adversary}



First, we introduce the adversary that adheres to \emph{dynamic diameter} $D$,
which gives a bound
on the duration of the information propagation from a \VSRC{}
to the entire network. We showed in \cite[Lem.~1]{SWS15:arxiv}
that always $D\leq n-1$; a priori restricting $D < n-1$ also allows
modelling dynamic networks where information propagation is guaranteed 
to be faster than in the worst case (as in expander graphs \cite{BRSSW16:TR}, 
for example).

\begin{definition}[Dynamic diameter D]\label{def:dynamic-diameter}
$\DBMA$ is the message adversary that guarantees dynamic (network)
diameter $D$, i.e., for all graph sequences $\sigma \in \DBMA$,
for all subsequences
$\left(\G^{r_1}, \ldots, \G^{r_1+D-1} \right) \subset \sigma$
of $D$ consecutive \Rooted{R} communication graphs,
we have $R \subseteq \CP{p}{r_1-1}{r_1+D-1}$ for every $p\in\Pi$.
\end{definition}

The following liveness property, \emph{eventual stability},
ensures that eventually every graph sequence $\sigma$ has
a \Rooted{R} subsequence $\sigma' \subseteq \sigma$ of length $x$.
Here $\Sigma$ denotes the unrestricted message adversary, i.e., the set of
all communication graph sequences.

\begin{definition}
  $\BLiveness(x) := \{ \sigma \in \Sigma \mid \exists R \subseteq \Pi \colon$
   some  $\sigma' \subseteq \sigma$ with $|\sigma'| \geq x$ is
  $R$-rooted $\}$.
  \label{def:liveness}
\end{definition}

For finite $x$, $\BLiveness(x)$ alone is insufficient for solving consensus:
Arbitrarily long sequences of graphs that are not rooted
before the stability phase occurs can fool any consensus algorithm to make
wrong decisions. For this reason, we introduce a safety property 
in the form of the message adversary that generates only rooted graphs.

\begin{definition}
  $\BSafety := \{ \sigma \in \Sigma \mid$
    every $\Gr$ of $\sigma$ is rooted $\}$.
\end{definition}


The short-lived eventually stabilizing message adversary $\BStable(D+1)$
used throughout the main part of our paper
adheres to the dynamic diameter $D$, guarantees that every $\Gr$ is
rooted and that every sequence has a subsequence of at least $x=D+1$
consecutive communication graphs with a \VSRC.
Since processes are aware under which adversary they are, they have common
knowledge of the dynamic diameter $D$ and the duration of the stability phase $x$.

\begin{definition}\label{def:bstableshort}
  We call
$\BStable(x) = \BSafety \cap \BLiveness(x) \cap \DBMA$
  the short-lived eventually stabilizing message adversary with stability period $x$.
\end{definition}

We observe that 
$\BLiveness(x) \supseteq \BLiveness(D)$ for any $1 \leq x \leq D$, 
hence it follows that $\BStable(x) \supseteq \BStable(D)$.
This simple set inclusion turns out to be quite useful for the next section.

\section{Impossibility Results and Lower bounds} \label{sec:impossibility}


Even though processes know the dynamic diameter $D$, for very
short stability periods, this is not enough to solve consensus.
In \theoremref{thm:uniform-impossibility}, we prove
that if processes do not have access to
an upper bound on $n$ (some $\nbound$ with $\nbound \geq n$),
i.e., when the algorithm is uniform,
solving consensus is impossible if the period $x$ of eventual stability
is shorter than $2D$:
Here, processes can never be sure
whether a \VSRC{} occurred for at least $D$ rounds, albeit
detecting such a \VSRC{} is necessary to satisfy validity.

\begin{figure*}[h]
  \centering
  \scalebox{\impscale}{
    \begin{subfigure}{0.05\linewidth}
      \scalebox{\impscale}{
        \begin{tikzpicture}
          \node at (0, 0.5) {$\sigma_1$:};
          \BBExecLabel
        \end{tikzpicture}
      }
    \end{subfigure}
    \begin{subfigure}{.26\linewidth}
      \scalebox{\impscale}{
        \begin{tikzpicture}[tight background]
          \node[shift={(0,1)}] (p1) at (0:0) {$p_1$};
          \node (p2) at (1,1) {$p_2$};
          \node[shift={(1,-0.6)}] (pD) at (0:0) {$p_D$};
          \node (pD1) at (0,-0.6) {$p_{D+1}$};
          \node (pD2) at (2,-0.6) {$p_{D+2}$};

          \node[rotate=90] (dots) at (1,0.2) {\dots};
          \draw[->] (p2) -- (dots);
          \draw[<-] (pD) -- (dots);

          \draw[->] (p1) -- (p2);
          \draw[->] (p2) -- (dots);
          \draw[<-] (pD) -- (dots);
          \draw[->] (pD) -- (pD1);
          \draw[->] (pD) -- (pD2);

          \node (left-paren)  at ( -.6, .2) {$\left( \vphantom{\rule{1pt}{32pt}} \right.$};
          \node (right-paren) at ( 2.85, .2) {$\left. \vphantom{\rule{1pt}{32pt}}
          \right)_{1}^{2D-1}$};
          \BBCommGraph
        \end{tikzpicture}
      }
    \end{subfigure}
    \begin{subfigure}{.26\linewidth}
      \scalebox{\impscale}{
        \begin{tikzpicture}[tight background]
          \node (p1) at (0,-0.6) {$p_{D+1}$};
          \node (pD) at (1,1) {$p_1$};
          \node (p2) at (0,1) {$p_{D+2}$};
          \node (pD1) at (1,-0.6) {$p_{D}$};

          \node[rotate=90] (dots) at (1,0.2) {\dots};
          \draw[->] (pD) -- (dots);
          \draw[<-] (pD1) -- (dots);

          \draw[->] (p1) -- (p2);
          \path (p2) edge[lossy, ->, bend left] (p1);
          \draw[->] (p2) -- (pD);

          \node (left-paren)  at ( -.6, .2) {$\left( \vphantom{\rule{1pt}{32pt}} \right.$};
          \node (right-paren) at ( 1.6, .2) {$\left. \vphantom{\rule{1pt}{32pt}}
          \right)_{2D}^{\infty}$};
          \BBCommGraph
        \end{tikzpicture}
      }
    \end{subfigure}
    \begin{subfigure}{.26\linewidth}
      \begin{tikzpicture}[tight background]
        \BBCommGraph
      \end{tikzpicture}
    \end{subfigure}
    \begin{subfigure}{.26\linewidth}
      \begin{tikzpicture}[tight background]
        \BBSingleGraph
      \end{tikzpicture}
    \end{subfigure}
  }
  
  \scalebox{\impscale}{
    \begin{subfigure}{.05\linewidth}
      \scalebox{\impscale}{
        \begin{tikzpicture}
          \node at (0,0.5) {$\sigma_2$:};

          \BBExecLabel
        \end{tikzpicture}
      }
    \end{subfigure}
    \begin{subfigure}{.26\linewidth}
      \scalebox{\impscale}{
        \begin{tikzpicture}[tight background]
          \node[shift={(0,1)}] (p1) at (0:0) {$p_1$};
          \node (p2) at (1,1) {$p_2$};
          \node[shift={(1,-0.6)}] (pD) at (0:0) {$p_D$};
          \node (pD1) at (0,-0.6) {$p_{D+1}$};
          \node (pD2) at (2,-0.6) {$p_{D+2}$};
          \node (pn) at (2,1) {$p_{n}$};

          \node[rotate=90] (dots) at (1,0.2) {\dots};
          \draw[->] (p2) -- (dots);
          \draw[<-] (pD) -- (dots);

          \node[rotate=90] (dots) at (2,0.2) {\dots};
          \draw[->] (pD2) -- (dots);
          \draw[<-] (pn) -- (dots);

          \draw[->] (p1) -- (p2);
          \draw[->] (pD) -- (pD1);
          \draw[->] (pD) -- (pD2);

          \node (left-paren)  at ( -.6, .2) {$\left( \vphantom{\rule{1pt}{32pt}} \right.$};
          \node (right-paren) at ( 2.85, .2) {$\left. \vphantom{\rule{1pt}{32pt}}
          \right)_{1}^{D-1\phantom{2}}$};
          \BBCommGraph
        \end{tikzpicture}
      }
    \end{subfigure}
    \begin{subfigure}{.28\linewidth}
      \scalebox{\impscale}{
        \begin{tikzpicture}[tight background]
          \node[shift={(0,1)}] (p1) at (0:0) {$p_1$};
          \node (p2) at (1,1) {$p_2$};
          \node[shift={(1,-0.6)}] (pD) at (0:0) {$p_D$};
          \node (pD1) at (0,-0.6) {$p_{D+1}$};
          \node (pD2) at (2,-0.6) {$p_{D+2}$};
          \node (pn) at (2,1) {$p_{n}$};

          \node[rotate=90] (dots) at (1,0.2) {\dots};
          \draw[->] (p2) -- (dots);
          \draw[<-] (pD) -- (dots);

          \node[rotate=90] (dots) at (2,0.2) {\dots};
          \draw[->] (pD2) -- (dots);
          \draw[<-] (pn) -- (dots);

          \draw[->] (p1) -- (p2);
          \path (p2) edge[lossy, ->, bend left] (p1);
          \draw[->] (pD) -- (pD1);
          \draw[->] (pD) -- (pD2);

          \node (left-paren)  at ( -.6, .2) {$\left( \vphantom{\rule{1pt}{32pt}} \right.$};
          \node (right-paren) at ( 2.90, .2) {$\left. \vphantom{\rule{1pt}{32pt}}
          \right)_{D}^{2D-1}$};
          \BBCommGraph
        \end{tikzpicture}
      }
    \end{subfigure}
    \begin{subfigure}{.25\linewidth}
      \scalebox{\impscale}{
        \begin{tikzpicture}[tight background]
          \node (p1) at (0,-0.6) {$p_{D+1}$};
          \node (pD) at (1,1) {$p_1$};
          \node (p2) at (0,1) {$p_{D+2}$};
          \node (pD1) at (1,-0.6) {$p_{D}$};
          \node (pD2) at (2,-0.6) {$p_{D+3}$};
          \node (pn) at (2,1) {$p_{n}$};

          \node[rotate=90] (dots) at (1,0.2) {\dots};
          \draw[->] (pD) -- (dots);
          \draw[<-] (pD1) -- (dots);

          \node[rotate=90] (dots) at (2,0.2) {\dots};
          \draw[->] (pD2) -- (dots);
          \draw[<-] (pn) -- (dots);

          \draw[->] (p1) -- (p2);
          \path (p2) edge[lossy, ->, bend left] (p1);
          \draw[->] (p2) -- (pD);
          \draw[->] (pD1) -- (pD2);

          \node (left-paren)  at ( -.6, .2) {$\left( \vphantom{\rule{1pt}{32pt}} \right.$};
          \node (right-paren) at ( 2.7, .2) {$\left. \vphantom{\rule{1pt}{32pt}}
          \right)_{2D}^{\tau}$};
          \BBCommGraph
        \end{tikzpicture}
      }
    \end{subfigure}
    \begin{subfigure}{.26\linewidth}
      \scalebox{\impscale}{ 
        \begin{tikzpicture}[tight background]
          \node[shift={(0.5,0.2)}] (p1) at (0:0) {$p_n$};

          \node (left-paren)  at ( 0, .2) {$\left( \vphantom{\rule{1pt}{32pt}} \right.$};
          \node (right-paren) at ( 1.2, .2) {$\left. \vphantom{\rule{1pt}{32pt}}
          \right)_{\tau+1}^{\infty}$};
          \BBSingleGraph
        \end{tikzpicture}
      }
    \end{subfigure}
  }
  \caption{
    Communication graph sequences of \theoremref{thm:uniform-impossibility}.
    A dotted edge
    represents and edge which is in $\G^i$ if and only if it is not in $\G^{i-1}$. 
    We assume there is an edge from every process depicted in the graph to every
    process not depicted in the graph.
  }
  \label{fig:D-uniform-imposs}
\end{figure*}
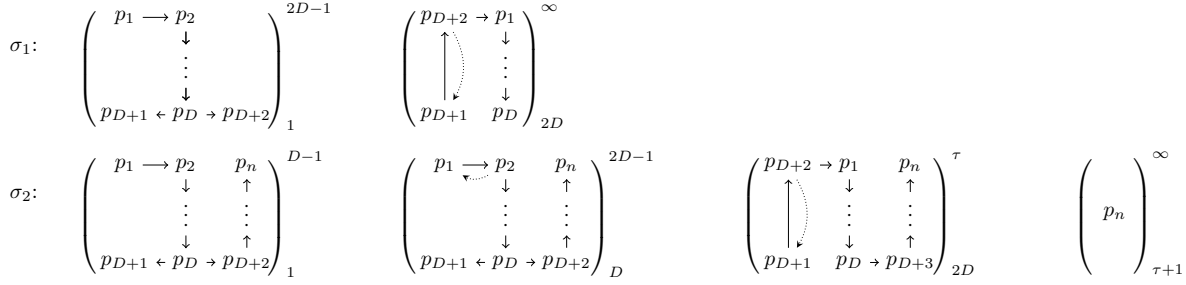

\begin{theorem} \label{thm:uniform-impossibility}
  There is no uniform consensus algorithm for $\BStable(x)$
  with $0 < x < 2D$.
\end{theorem}
\begin{proof}
  As, for $x > x'$, it holds that $\BLiveness(x) \subset \BLiveness(x')$,
  it suffices to show that consensus is impossible under \MA{}
  $\MAD =  \BSafety \cap \BLiveness(2D-1) \cap \DBMA$.

  Suppose some algorithm $\Alg$ solves consensus
  under $\MAD$.
  We provide two admissible executions $\epsilon_1, \epsilon_2$
  (based on $\sigma_1$, resp.\ $\sigma_2$, from \figureref{fig:D-uniform-imposs})
  of $\Alg$ where $\epsilon_1 \indist{p_{D+1}} \epsilon_2$.
  We show that $p_{D+1}$ decides $0$ in $\epsilon_1$ and,
  for almost all values of $n$, process $p_n$ decides $1$ in $\epsilon_2$.
  
  Let $C^0$ be the initial configuration with input values
  $x_1 = \ldots = x_{D+2} = 0$ and $x_{D+3} = \ldots = x_n = 1$.

  Consider execution $\epsilon_1 = \Exec{C^0, \sigma_1}$
  with $\sigma_1$ from \figureref{fig:D-uniform-imposs}, where
  a dotted edge exists only in every second graph in a sequence, and all
  processes not depicted have an in-edge from every depicted process.
  $\sigma_1 \in \MAD$, since it guarantees eventual stability for $2D-1$
  rounds, adheres to the dynamic diameter $D$ and in every round the
  communication graph is rooted.
  By the assumed correctness of $\Alg$, there is a round $\tau$ by which
  every process has decided in $\varepsilon_1$.
  The decision must be $0$ because $p_{D+1}$ only ever saw processes that knew
  of input value $0$.  This is indistinguishable for $p_{D+1}$ from the
  execution where all processes did indeed start with input $0$, which,
  according to the validity property of consensus, implies a decision on $0$.

  Now, consider the execution $\varepsilon_2 = \Exec{C^0, \sigma_2}$ with
  $\sigma_2$ from \figureref{fig:D-uniform-imposs}.
  Again, $\sigma_2 \in \MAD$, since $p_n$ is a \VSRC{} for $r\geq \tau+1$.
  In every round $r \leq \tau$, $p \in \{ p_{D+1}, p_{D+2} \}$ have the
same view in $\varepsilon_1$ and $\varepsilon_2$: This is immediately
obvious for $1 \leq r \leq D-1$. For $D \leq r \leq 2D-1$, the
view of $p_1$ is different, but this difference is not revealed
to $p$ by the end of round $2D-1$. Finally, in rounds $2D \leq r \leq \tau$,
the processes $\{ p_{D+1}, p_{D+2} \}$ hear only from themselves in
both executions, hence maintain $\epsilon_1 \indist{p_{D+1}} \epsilon_2$.

  Consequently, by round $\tau$, $p_{D+1}$ has decided $0$.
  Yet, in executions where $n > \tau + D + 3$, according to $\epsilon_2$ in
  \figureref{fig:D-uniform-imposs}, we have that $p_n$ never saw a process that had an
  input value different from $1$.
  By validity and an analogous argument as above, $p_n$ must hence decide $1$
  in $\epsilon_2$ here, which provides the required contradiction.
\end{proof}

As our next result, we present a lower bound for the duration 
$x$ of the stable period: We prove that even 
in the non-uniform case, consensus is impossible under
$\BStable(x)$ if $x \leq D$ (\theoremref{thm:D-imposs}).
Note that this result improves the lower bound $x \geq D-1$ established
in \cite{BRS12:sirocco} and thus reveals that the latter was not tight.

Our lower bound can be seen as a generalization of the ``lossy-link''
impossibility from \cite[Theorem 2]{SWK09}, a particular formalization
of consensus for two processes.
There, it was shown that consensus is impossible in a two-process
system where all messages except one may get lost in every round.
In our terminology, this means that, for $n=2$ and $D=1$, consensus is
impossible under $\BStable(1)$, even if processes are aware of the
size of the system. For general $D$, \theoremref{thm:D-imposs} below shows that
a stability period of $D$ or less rounds is insufficient
for solving consensus for any value of $\nbound$ as well.
Informally, the reason is that there are executions where,
even with a stability phase of $D$ rounds,
some process cannot precisely determine the
root component of the stability window.
In \figureref{fig:D-imposs}, for instance, $p_{D+1}$ cannot determine whether
$p_1$ alone or $p_1$ and $p_2$ together constitute the root component after a
sequence of $D$ successive occurrences of $\G_a$, respectively, $\G_b$.
The determination of this root
component is crucial, however, since any root component could be the ``base'' for a
decision in the suffix of some indistinguishable execution.
\iftoggle{TR}{}{The formal proof can be found in the appendix.}

\iftoggle{TR}{
  Our impossibility proof relies on a bivalence argument:
Consider some algorithm $\Alg$ that solves the binary consensus problem, where, 
for every process $p$, the initial value $x_p \in \{0, 1 \}$.
Given some message adversary $\MAD$,
in analogy to \cite{FLP85}, we call a configuration
$C = \Exec{C^0, \sigma}$ of $\Alg$
\emph{univalent} or, more specifically,
\emph{$v$-valent}, if all processes 
decided $v$ in $\Exec{C, \sigma'}$ for any
$\sigma'$ where $\sigma \sigma' \in \MAD$.
We call $C$ \emph{bivalent}, if it is not univalent.

Before stating our main theorem, we need to establish an essential
technical lemma.
It shows for $n>2$ that by adding/removing a single edge at a time, we can
arrive at any desired rooted communication graph 
when starting from any other rooted communication graph.
Furthermore, during this construction, we can avoid any graphs 
that contain a certain ``undesirable'' root component $R''$.

\begin{lemma} \label{lem:increment-root-grid}
  Let $n>2$, $\G$ be a rooted communication graph with $\Root(\G) = \{ R \}$,
  $\G'$ be a rooted communication graph with $\Root(\G') = \{ R' \}$, and $R''$ be
  some root component with $R'' \neq R$ and $R'' \neq R'$.
  Then, there is a sequence of communication graphs,
  $\G = \G_1, \ldots, \G_k = \G'$ s.t.\ each $\G_i$ of the sequence
  is rooted, $\Root(\G_i) \neq R''$, and,
  for $1 \leq i < k$, $\G_i$ and $\G_{i+1}$ differ only in a single edge.
\end{lemma}

\begin{proof}
  We show that for any rooted communiaction graph $\incG$ with $\Root(\G_i) = \incR$,
  there is such a sequence $\incG = \incG_1, \ldots, \incG_j = \incG'$
  for any communication graph $\incG'$ with $\Root(\incG')=\incR'$
  if $\incR'$
  differs from $\incR$ in at most one process, i.e.,
  $\lvert \incR' \cup \incR \setminus \incR' \cap \incR \rvert \leq 1$.
  Repeated application of this fact implies the lemma, because for $n>2$ we can always
  find a sequence $R = R_1, \ldots, R_l = R'$ of subsets of $\Pi$  s.t.\ 
  for
  each $R_i$ of the sequence we have $R_i \neq R''$ and, for $1 \leq i < l$,
  $R_i$ differs from $R_{i+1}$ by exactly one process.
  To see this, we observe that in the Hasse diagram of the power set of $\Pi$,
  ordered by set inclusion, there are always two upstream paths leading from any 
two subsets of $\Pi$ to a common successor.

  We sketch how to construct the desired communication graphs $\incG_i$ of the sequence
  in three phases.

  \emph{Phase 1:} Remove all edges (one by one) between nodes of $\incR$ until only a cycle (or, in general, a 
circuit) remains,
  remove all edges between nodes outside of $\incR$ until only chains going out from
  $\incR$ remain.

  \emph{Phase 2:} If we need to add a node $p$ to $\incR = \Root(\G_i)$ to arrive
at $\incR'$, for some $q \in \incR$,
  first
  add $(q \ra p)$.
  For any $q' \neq q$ where $(q' \ra p) \in \incG_i$ with
  $p \neq q'$,
  remove $(q' \ra p)$.
  Finally, add $(p \ra q)$.

  If we need to remove a node $p$ from $\incR$ to arrive
  at $\incR'$, for any
  $(q \ra p)$, $(p \ra q') \in \incG_i$,
  with $q, q' \in \incR$
  subsequently
  add $(q \ra q')$ and $(q' \ra q)$, then
  remove $(p \ra q)$ and $(p \ra q')$.
  Note that we perform this step also when $q = q'$.

  \emph{Phase 3:} Since we now already have some communication graph $\incG_i$ with $\Root(\incG_i) = \incR'$,
  it is easy to add/remove edges one by one to arrive at the topology of $\incG'$.
  First, we add edges until the nodes of $\incR'$ are completely connected
  among each other, the nodes not in $\incR'$ are completely connected among
  each other, and there is an edge from every node of $\incR'$ to each node not in $\incR'$.
  Second, we remove the edges not present in $\incG'$.
\end{proof}

}{}

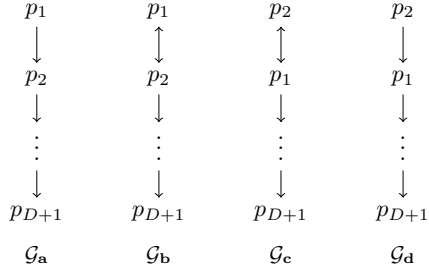
\begin{figure}[h]
  \centering
  \captionsetup[subfigure]{justification=centering}
  \scalebox{\impscale}{
    \begin{subfigure}[t]{.2\linewidth}
      \centering
      \begin{tikzpicture}

        \node (p1) at               ( 0,  0) {$p_1$};
        \node (p2) at               ( 0, -1) {$p_2$};
        \node[rotate=90] (dots) at  ( 0, -2.0) {\dots};
        \node (pD) at               ( 0, -3) {$p_{D+1}$};

        \draw[->] (p1) to (p2);
        \draw[->] (p2) to (dots);
        \draw[->] (dots) to (pD);
      \end{tikzpicture}
      \caption*{$\mathbf{\G_a}$}
      \label{fig:D-imposs-a}
    \end{subfigure}
    \begin{subfigure}[t]{.2\linewidth}
      \centering
      \begin{tikzpicture}

        \node (p1) at               ( 0,  0) {$p_1$};
        \node (p2) at               ( 0, -1) {$p_2$};
        \node[rotate=90] (dots) at  ( 0, -2.0) {\dots};
        \node (pD) at               ( 0, -3) {$p_{D+1}$};

        \draw[<->] (p1) to (p2);
        \draw[->] (p2) to (dots);
        \draw[->] (dots) to (pD);
      \end{tikzpicture}
      \caption*{$\mathbf{\G_b}$}
      \label{fig:D-imposs-b}
    \end{subfigure}
    \begin{subfigure}[t]{.2\linewidth}
      \centering
      \begin{tikzpicture}

        \node (p2) at               ( 0,  0) {$p_2$};
        \node (p1) at               ( 0, -1) {$p_1$};
        \node[rotate=90] (dots) at  ( 0, -2.0) {\dots};
        \node (pD) at               ( 0, -3) {$p_{D+1}$};

        \draw[<->] (p1) to (p2);
        \draw[->] (p1) to (dots);
        \draw[->] (dots) to (pD);
      \end{tikzpicture}
      \caption*{$\mathbf{\G_c}$}
      \label{fig:D-imposs-c}
    \end{subfigure}
    \begin{subfigure}[t]{.2\linewidth}
      \centering
      \begin{tikzpicture}

        \node (p2) at               ( 0,  0) {$p_2$};
        \node (p1) at               ( 0, -1) {$p_1$};
        \node[rotate=90] (dots) at  ( 0, -2.0) {\dots};
        \node (pD) at               ( 0, -3) {$p_{D+1}$};

        \draw[<-] (p1) to (p2);
        \draw[->] (p1) to (dots);
        \draw[->] (dots) to (pD);
      \end{tikzpicture}
      \caption*{$\mathbf{\G_d}$}
      \label{fig:D-imposs-d}
    \end{subfigure}
  }
  \caption{
    Communication graphs for \theoremref{thm:D-imposs}.
    We assume there is an edge from every process depicted in the graph to every
    process not depicted in the graph.
  }
  \label{fig:D-imposs}
\end{figure}
\begin{theorem}
  \iftoggle{APPENDIX}{}{
    \label{thm:D-imposs}}
  There is no 
  non-uniform consensus algorithm
for \MA{} $\BStable(x)$
  with $1 \leq x \leq D$. This holds even if the adversary must guarantee 
  that, in addition to $\BStable(x)$, the first $D$ rounds are \Rooted{R} (provided the 
processes do not know $R$ a priori).
\end{theorem}

\iftoggle{TR}{
  \begin{proof}
  In the case where $D=1$, we need to show the impossibility of $\BStable(1)$.
  We immediately note that $\sigma \in \BStable(1)$ if and only if we have that each $\G \in \sigma$
  is rooted and a has graph diameter of $1$.
  Clearly, the graph sequence where (i) two fixed processes $p,q$ have 
non-self-loop in-edges at most from
  each other and at least one of those present in every $\Gr$, and (ii) all
  other processes have an in-edge from both $p$ and $q$ is in $\BStable(1)$.
  However, a trivial modification of \cite[Theorem 2]{SWK09} shows that 
consensus (in particular, among $p$ and $q$) is impossible in this setting.
  
  For the remainder of the proof, let us thus assume $D \geq 2$.
  Since $\BStable(x) \supset \BStable(D)$ for $x \leq D$, it suffices
  to show the impossibility of consensus under $\BStable(D)$:
  The execution $\varepsilon$ where consensus cannot be solved is
  admissible under $\BStable(x)$ if it is admissible under
  $\BStable(D)$.
  The proof proceeds roughly along the lines of \cite[Lemma 3]{SWK09}.
  It first shows that there is a bivalent round-$D$ configuration for any consensus
  algorithm and proceeds to show by induction that every bivalent configuration
  has a bivalent successor configuration.
  Hence, any consensus algorithm permits a
  perpetually bivalent execution under $\BStable(D)$,
  where consensus cannot be solved.
  
  We show that a bivalent execution is even contained in the adversary
  $\BStable'(D)\subseteq \BStable(D)$, which consists of those executions
  of $\BSafety \cap \DBMA$ where already the first $D$ rounds are
  \Rooted{R}.

  For the induction base, we show that not all round $D$
  configurations of $\Alg$ can be univalent:
  Assume that some algorithm $\Alg$ solves consensus under $\BStable'(D)$ and
  suppose that all round $D$ configurations of $\Alg$ were univalent.

  Let $C^0$ be some initial configuration of $\Alg$ with $x_{p_1} = 0$ and $x_{p_2} = 1$
  and recall the graphs $\G_a, \G_b, \G_c$ and $\G_d$ from \figureref{fig:D-imposs}.
  For $i \in \{a, b, c, d \}$ let $C^D_i = \Exec{C^0, (\G_i)_{1}^D}$
  denote the configuration which results from applying $\G_i$ $D$ times to $C^0$.
  Let $\Star{p}$ denote the star-like graph where there is an edge from the
  center vertex $p$ to every other vertex and from every vertex to itself but
  there are no other edges in the graph.
  Clearly, $C^D_a$ is $0$-valent since $\Exec{C^D_a, \SeqG{\Star{p_1}}} \in \BStable'(D)$
  and for $p_1$ this is indistinguishable from the situation where all processes $p$ have
  $x_p = 0$.
  A similar argument shows that $C^D_d$ is $1$-valent.
  
  Consider two cases:

  (1) $C^D_b$ is $1$-valent.
  But then, $C^D_a$ cannot be $0$-valent since
  $\Exec{C^D_a, \SeqG{\Star{p_{D+1}}}} \indist{p_{D+1}} \Exec{C^D_b, \SeqG{\Star{p_{D+1}}}}$.
  
  (2) $C^D_b$ is $0$-valent.
  Then, $C^D_c$ is also $0$-valent since
  $\Exec{C^D_b, \SeqG{\Star{p_{1}}}} \indist{p_{1}} \Exec{C^D_c, \SeqG{\Star{p_{1}}}}$.
  But then $C^D_d$ cannot be $1$-valent because 
  $\Exec{C^D_c, \SeqG{\Star{p_{D+1}}}} \indist{p_{D+1}} \Exec{C^D_d, \SeqG{\Star{p_{D+1}}}}$.

  Hence, not all round $D$ configurations are univalent.


  For the induction step, let us assume that there exists a bivalent round $r$ configuration
  $C^r$ at the end of round $r \geq D$.
  For a contradiction, assume that all round $r+1$ configurations reachable from $C^r$
  are univalent.
  Thus, there exists a $0$-valent round $r+1$ configuration $C_0^{r+1} = \Exec{C^r, \G_0}$
  that results from applying some communication graph $\G_0$ to $C^r$.
  Moreover, there is a $1$-valent round $r+1$ configuration $C_1^{r+1} = \Exec{C^r, \G_1}$
  that results from applying some communication graph $\G_1$ to $C^r$.
  
  First, let us show that for $\G \in \{ \G_0, \G_1 \}$, it holds that, if $\Root(\G) = \Root(\G^r)$,
  there is an applicable graph $\G'$ s.t.\ $\Exec{C^r, \G'}$ has the same valency as $\Exec{C^r, \G}$
  and $\Root(\G) \neq \Root(\G')$.
  The reason for this is that we can construct $\G'$ from $\G$ by simply adding an edge $(p \ra q)$
  for some $q \neq p$, $p \not\in \Root(\G)$, $q \in \Root(\G)$ if $|\Root(\G)| = 1$, respectively,
  by removing $(p \ra q)$ for some $q \in \Root(\G)$ and all $p \neq q$ if $|\Root(\G)| > 1$.
  This yields a graph $\G'$ with the desired property because
  $\Exec{C^r, \G, \SeqGr{\Star{p}}} \indist{p}
  \Exec{C^r, \G', \SeqGr{\Star{p}}}$.
  The applicability of $\G'$ follows because $\G'$ is rooted and
  $\Root(\G') \neq \Root(\Gr)$ ensures that the resulting subsequence is a prefix of some
  sequence of $\DBMA$ for any $D>1$, because, for these choices of $D$, a changing root
  component trivially satisfies \defref{def:dynamic-diameter}.

  Hence we can find graphs $\G'_0, \G'_1$ such that $\Root(\G'_0) \neq \Root(\Gr)$,
  $\Root(\G'_1) \neq \Root(\Gr)$, and $\Exec{C^r, \G'_0}$ is $0$-valent
  while $\Exec{C^r, \G'_1}$ is $1$-valent.
  As we assumed $D \geq 2$ it follows that $n > 2$.
  We can hence apply \lemmaref{lem:increment-root-grid} to go from $\G'_0$ to $\G'_1$
  by adding/removing a single edge at a time,
  without ever arriving at a graph that has more than one root component or has the same root component as
  $\Gr$.
  Somewhere during adding/removing a single edge, we transition from a graph $\G_i$ to 
  a graph $\G_{i+1}$, by modifying an edge $(p \ra q)$,
  where the valency of $C = \Exec{C^r, \G_i}$ differs from the valency of
  $C' = \Exec{C^r, \G_{i+1}}$.
  Nevertheless, $\G_i$ and $\G_{i+1}$, are applicable to $C^r$ because they are
  rooted and have a different root component as $\G^r$, hence guarantee the membership of
  the sequence in $\DBMA$ for any $D > 1$.
  However, $C$ and $C'$ cannot have a different valency because
  $\Exec{C, \SeqGr{\Star{p}}} \indist{p} \Exec{C', \SeqGr{\Star{p}}}$.
  This is a contradiction and hence not all round $r+1$ configurations
  can be univalent.
\end{proof}

}{}

\tikzset{                                                                       
   >=latex,shorten <=2pt,shorten >=2pt,                                          
   inner sep=1pt,                                                                
   dot/.style={->},                                                              
   dash/.style={<->,dashed},                                                     
   proc/.style={}
}          

\section{Solving Consensus with $D+1$ Rounds of Stability} \label{sec:algorithm}


We now present \algref{alg:consensus}, which, under the message adversary
$\BStable(D+1)$, solves consensus if $D$ and a bound $N \geq n$ is known a
priori. 
It relies on an underlying
graph approximation algorithm used already in 
\cite{SWSBR15:NETYS,SWS15:arxiv}, which provides each process with
a local estimate of past communication graphs. It
is easily implemented on top of any full-information protocol\footnote{Since 
we are mainly interested in the solvability aspect of consensus,
we consider a full-information protocol where processes forward
their entire state in every round.
As will become apparent from the correctness proof of \algref{alg:consensus}, however, 
it would in fact suffice to exchange the
relevant data structures for the last
$N(D+2N)$ rounds at most. Hence, the same engineering improvements as outlined
in \cite{SWS15:arxiv} can be used to get rid of the simplifying 
full-information-protocol assumption.} 
and
has been omitted for brevity.
As stated in Corollaries~\ref{cor:underapproximation} and~\ref{cor:interval-detection}
below, the
local graph estimates allow to faithfully
detect the existence of root components with a latency of at most
$D$ rounds,
with some restrictions.

In our consensus algorithm, processes operate in two alternating
states (``not locked'' and ``locked''). 
Basically, in the ``not locked'' state, process $p$ tries to find a root 
component and, if successful, locks on to it,
i.e., considers it a potential base for its decision.
This is realized by adapting the local $\Proposal_p$ value 
to the value of the root component (which is the maximum of the
proposal values of its members) and setting
$\Locked_p \gets \true$,
thereby entering the ``locked'' state.
Subsequently, $p$
waits for contradictory evidence for a period of time that
is long enough to guarantee that every process in the system has adapted to
the value of the root component locked-on by $p$.
If $p$ finds contradictory evidence, it backs off by leaving the ``locked''
state through setting $\Locked_p \gets \false$.

In more detail, $p$ changes its state from ``not locked'' to ``locked''
when it detects a root component
via \lineref{line:root-guard}
that was present $D$ rounds before the current round. 
In this case, $p$
enters the locked state by setting
$\Locked_p \gets \true$ and $\Proposal_p$ to $\Proposal'$, the maximum value
known to any process of the root component in round $r-D$, via
Lines~\ref{line:lock-root} and~\ref{line:locked-root}.
It furthermore records the current round as $\ell$, the so-called \emph{lock-round},
via \lineref{line:lockround}.
This is an optimistic mechanism, i.e., the process ``hopes'' that this
root component was already the stable root component promised by the adversary.
In this sense, the process starts collecting evidence, 
possibly contradicting its
hope, that stems from the lock-round itself or a more recent round.

If $p$ finds a root component when it is already in the
``locked'' state, i.e, when
it is already locked-on to some root component, it proceeds as follows:
First, $p$ checks whether it detected a relevant change in the members of the
root component since its lock-round and whether the maximum value, known to the
processes of the new root component at the time, is different
from its current proposal value.
If both checks succeed, there was enough instability for the process to conclude
that the currently locked-on root component cannot belong to the stable period; it
hence locks-on to the newly detected root component.
If the values remained the same,
it puts the current round in a \emph{queue} of candidate lock-rounds
in \lineref{line:queue-append}.
Later, if the process finds evidence that indeed contradicts its hope (as detailed
below), it will
pick the next lock round from the candidate queue via \lineref{line:queue-getLock}, and
tries to find contradictory evidence from this round on, thereby ensuring
that it cannot miss the promised stable root component.

When $p$ is still in a ``locked'' state after completing its root component detection,
it searches for contradictory evidence. More precisely, contradictory 
evidence means that $p$ learned of another process $q$
s.t.\ $q$ is not locked-on to a root component with the same value as $p$.
This is realized via the guard of \lineref{line:guard-backoff}.
If $p$ finds such evidence, it leaves its locked state and enters the ``not locked'' 
state again.
A decision occurs when $p$ has received no contradictory evidence for
such a long time that every other process has seen that
$p$ might be in this situation via \lineref{line:lock-hirsch}.
As we prove in \lemmaref{lem:agreement}
below,
it is guaranteed that if $p$ passes the guard in
\lineref{line:decide}, then every other process has $p$'s decision value as its 
local proposal value forever after.
Hence, all future decisions will be based on this value, which leads the system
to a safe configuration.

\begin{algorithm}[ht!]
  \small
  \caption{Consensus algorithm, code for process $p$}
  \label{alg:consensus}
  \DontPrintSemicolon
  \SetKwInput{Initialization}{Initialization}
  \SetKwInput{Transmit}{Transmit round $r$ messages}
  \SetKwInput{Compute}{Round $r$ computation}

  \nonl Let $\lastround_q$ be the last round $p$ heard from
  $q$ in round $r$, i.e.,
  $q \in \CP{p}{\lastround_q}{r}$ and $q \notin \CP{p}{\lastround_q +1}{r}$.

  \BlankLine

  \Initialization{}
  $\Proposal \gets x_p$ \tcc{initially, $p$ proposes own input}
  $\Locked \gets \true$ \tcc{$p$ starts `locked-on'}
  $\ell \gets 1$ \tcc{initialize lock-round to start round}
  $\Queue \gets \emptyset$ \tcc{we assume that $\max(\emptyset) = 0$}
  \BlankLine

  \Compute{}
  $R \gets \Root^r_p({r-D})$ \;
  \tcc{Find all roots that were detected since the (estimated) start of the stability phase}
  $T \gets \{ \Root^r_p(i)\mid  \max(\max(\Queue), \ell) -D \leq i \leq r-D \}$
  \label{line:T-construction}\;
  \tcc{For each process, determine its known states within the last $\nbound$ rounds:}
  $S \gets \{ q^s \mid q \in \CP{p}{r-\nbound}{r} \land s \in [r-N, \lastround_q] \}$ \;
  \tcc{Find the proposal values of the locked-on processes of $S$}
  $S' \gets \set{ \Proposal_q^s \mid q^s \in S \land \Locked_q^{s} = \true }$
  \label{line:Sprime-construction}\;
  \tcc{Collect all known states of the last $\nbound(D+2\nbound)$ rounds}
  $S'' \gets \set{ q^s \mid q \in \CP{p}{r-\nbound(D+2\nbound)}{r} \land s \in
	  [r - (D+2\nbound)^2, \lastround_q]}$
  \label{line:Sprimeprime-construction}\;
  \If{$R \neq \{ \bot \}$} { \label{line:root-guard}
  $\Proposal' \gets \max_{q \in R} \Proposal_q^{r-D}$ \label{line:calc-newProp} \;
  \If{$\Locked = \false$ or $(\Proposal' \neq \Proposal$ and 
    $\abs{T} > 1)$ \label{line:guard-newRoot}}{
      $\Proposal \gets \Proposal'$ \; \label{line:lock-root}
      $\Locked \gets \true$ \; \label{line:locked-root}
      $\ell \gets r$ \; \label{line:lockround}
    }
    \lElseIf{$\Proposal'=\Proposal$}{
      add $r$ to $\Queue$ \label{line:queue-append}
    }
  }
  \If{$r \geq \ell + N$ and $q^{s} \in S$ s.t.\ $\Locked_q^{s} = \false$
  or $\Proposal \neq \Proposal_q^{s}$ \label{line:guard-backoff}}{
    remove all $r'$ with $r' \leq s$ from $\Queue$ \label{line:queue-prune} \;
    \lIf{$\Queue \neq \emptyset$ \label{line:queue-getLock}}{
    $\ell \gets min(\Queue)$ 
	}
    \lElse{
      $\Locked \gets \false$ \label{line:locked-backoff}
    }
  }
  \lIf{$r \geq \ell + 2N$ and $S'$ contains single element $x \neq \Proposal$}
    {$\Proposal \gets x$} \label{line:lock-hirsch}
  \lIf {not yet decided, $r = \ell + \nbound(D + 2\nbound)$, and, for all $q^s$ of $S''$, $\Locked_q^s = \true$ and $\Proposal_q^s = \Proposal$}{decide $\Proposal$} \label{line:decide}

\end{algorithm}

\subsection*{Correctness proof\iftoggle{APPENDIX}{ of \algref{alg:consensus}}{}}
\newcommand{\mydatafn}[1]{\textsf{\small #1}}
\SetDataSty{mydatafn}

We now prove the correctness of \algref{alg:consensus} under $\BStable(D+1)$.
As shown in \cite[Lemmas 3 and 4]{SWS15:arxiv}, the
simple graph approximation algorithm running underneath our
consensus algorithm allows processes to faithfully
detect root components under certain circumstances.
Denoting process $p$'s round $r$ estimate of $\Root(\G^s)$ as
$\Root^r_p(s)$, the following two key features can be
guaranteed algorithmically (we assume that $\Root^r_p(s)$ returns
$\{ \bot \}$ if $p$ is unsure about its estimate):

\begin{corollary} \label{cor:underapproximation}
  If $\Root^r_p(s) \neq \{ \bot \}$, then $\Root(\G^s) = \Root^r_p(s)$.
  Furthermore, in round~$r$, process $p$ knows the round $s$ state $q^s$ 
for every $q \in \Root(\G^s)$.
\end{corollary}
We note that \corollaryref{cor:underapproximation} relies critically on the fact that the graph
approximation algorithm maintains an under-approximation of past
communication graphs.

\begin{corollary} \label{cor:interval-detection}
  If $\Seqr{s}{s+D}$ is \Rooted{R}, then we have $\Root^{s+D}_p(s) = R$
for every process $p$.
\end{corollary}
Note that \corollaryref{cor:interval-detection} 
is a consequence of the message adversary respecting the
dynamic diameter $D$: In round $s+D$, for each process $q \in R$,
every process received a message containing $q^s$,
the round $s$ state of $q$.

We are now ready to formally analyze \algref{alg:consensus}.
For this purpose, we introduce the useful term of \emph{$v$-locked
root component} to denote a root component where all members are
\emph{$v$-locked} for the same $v$, i.e., have the proposal $v$ and are 
in the locked state.

\begin{definition}
  The round $r$ state $p^r$ of $p$ is a $v$-locked state
  if $\Locked_p^r = \true \land \Proposal_p^r = v$.
  $R = \Root(\Gr)$ is a $v$-locked root component if
  all $q^r$ with $q \in R$ are $v$-locked.
  \label{def:v-locked-root}
\end{definition}

The following technical lemma is key for the proof of the agreement
property of consensus. It assures that if a sequence of at least $D+2N$
communication graphs occurs in which all root components happen to be
$v$-locked by our algorithm, then all the processes' proposal values 
are $v$ in all subsequent rounds.

\begin{lemma}
  Let $\sigma = \Seqr{a}{c}$ be a sequence of communication graphs such that,
for some fixed value $v$,  every $\Gr \in \sigma$ has a $v$-locked root component and
  $|\sigma| \geq D+2N$.
  Then, for any process $p$ and all rounds $r \geq c$, we have
  $\Proposal_p^r = v$.
  \label{lem:2N-v-locked}
\end{lemma}

\begin{proof}
  Let $b = a+D+N-1$ and note that $c \geq a + D + 2N -1 = b+N$.
  We prove our lemma using the invariants $\InvA, \InvB$ that are
  defined as follows:
  $\InvA_p(r) \Leftrightarrow \Locked_p^r = \false \lor \Proposal_p^r = v$,
  $\InvB_p(r) \Leftrightarrow \Proposal_p^r = v$.
  We use 
  $\InvA(r)$, resp.\ $\InvB(r)$, to denote that $\InvA_p(r)$, resp.\ $\InvB_p(r)$,
  holds for every $p \in \Pi$.

  First, we show $\InvA(r)$ for $r \in [b,c]$, by proving that for any process $p$,
  $\InvA_p(r)$ holds. We distinguish two cases.

  (1) If $p$'s lockround $\ell$, just
  before executing \lineref{line:guard-backoff} in round $r$,
  satisfies $\ell > r-N$,
  we have that $\Proposal_p^{\ell} = v$, because
  of \corollaryref{cor:underapproximation} and because
  $\Root(\G^{\ell-D})$ is $v$-locked.
  If $p$ modified $\Proposal_p$ in some round $s \in [\ell, r]$,
  it must have done so via \lineref{line:lock-root} or
  \lineref{line:lock-hirsch}.
  The former again means that $\Proposal_p^s = v$, due to
  \corollaryref{cor:underapproximation} and since
  $\Root(\G^{s-D})$ is $v$-locked by assumption.
  The latter means that the guard of \lineref{line:lock-hirsch}
  was passed and hence $s \geq \ell + 2N > b+N$,
  which contradicts that $s \in [\ell, r] \subseteq [b-N+1, c]$.

  (2) If $\ell \leq r- \nbound$, i.e., $r \geq \ell + \nbound$,
  because we assume that $\Root(\G^i)$ is $v$-locked
  for each $\G^i$ of $(\G^i)_{i=a}^r$ and $r \in [b, c] \Rightarrow r \geq a+N$,
  it follows from \lemmaref{lem:information-propagation} that
  there is a process $q$ and a round
  $s \in [r-\nbound, r]$
  s.t.\
  $q \in \Root(\G^s)$, $q \in \CP{p}{s}{r}$,
  and $q^s$ is $v$-locked.
  Since $\Root(\G^s)$ is $v$-locked by assumption,
  $S$ contains some state $q^s$ with $\Proposal_q^s = v$ and
  $\Locked_q^s = \true$.
  Therefore, $v \in S'^r$.
  It follows that if $S'^r$ contains a single element $x$,
  we have $x=v$.
  Since \lineref{line:lock-hirsch} is the only place in the code,
  apart from \lineref{line:lock-root}, where $p$ assigns
  a value to $\Proposal_p^r$, if $\Proposal_p^{r-1} = v$, then
  $\Proposal_p^r = v$.
  On the other hand, if $\Proposal_p^{r-1} \neq v$, because $r \geq \ell + \nbound$,
  $p$ evaluates the guard of \lineref{line:guard-backoff} to true.
  If $\Queue = \emptyset$, $p$ 
  executes \lineref{line:locked-backoff} and sets
  $\Locked_p^r \gets \false$.
  If $\Queue \neq \emptyset$, let $\ell' =min(\Queue)$.
  Due to \lineref{line:queue-prune}, we have $\ell' > s$,
  and because $s\in[r-\nbound,r]$, it follows that $\ell'>r-\nbound>b-\nbound>a+D-1$.
  Hence $\ell'-D \in [a,c-D] \subseteq [a,c]$.
  This means that during round $\ell'$, directly after executing
  \lineref{line:calc-newProp}, $p$ already had
  $\Locked = \true \land \Proposal = \Proposal'$.
  But according to \corollaryref{cor:underapproximation},
  this implies that $\Proposal = v$,
  because $\Root(\G^{\ell'-D})$ is $v$-locked by assumption during the
  interval $[a,c]$.

  We are now ready to complete the proof by using
  induction to show that $\InvB(r)$ holds for each $r \geq c$.
  For the base case, we prove $\InvB_p(c)$ for an arbitrary process $p$:

  If the lockround $\ell$ of $p$ in round $c$, just before reaching line 15,
  satisfies $\ell > c- 2N$, we can use the same arguments as in (1) above:
It follows from the assumption that $\Root(\G^i)$
  for $i \in [c-2N-D+1, c] \subseteq [a, c]$
  is $v$-locked and \corollaryref{cor:underapproximation} that whenever $p$
changes its proposal value in some round $s$, we have
  $\Proposal_p^s = v$. Consequently, $\Proposal_p^c = v$ as well.

  If $\ell \leq c-2\nbound$, we are in a similar situation as in (2)
above as $v \in S_p'^c$.
  Since $\InvA_p(s)$ holds for $s \in [b, c]$, $\Locked_q^s = \true$
  implies that $\Proposal_q^s = v$.
  Hence, there can be no $x \neq v$ with $x \in S_p'^c$.
  Therefore, $p$ executes \lineref{line:lock-hirsch} in round $c$ and
  sets $\Proposal \gets v$, hence $\Proposal_p^c = v$.

  For the induction step, we show that, for any process $p$,
  $\bigwedge_{i=c}^{r} \InvB(i)
  \Rightarrow \InvB_p(r+1)$.
  We distinguish three cases:

  (1) $p$ assigns a value to $\Proposal$ in round $r+1$ in
  \lineref{line:lock-root}.
  Irrespectively of $r$,
  we have $\Proposal^{r+1}_p = v$:
  If $r < c+D$, this holds
  because the root of $\G^{r+1-D}$ is $v$-locked.
  If $r \geq c+D$, we get the result from the fact that
  $\InvB(r-D)$ holds.

  (2) $p$ assigns a value to $\Proposal$ in round $r+1$ in
  \lineref{line:lock-hirsch}.
  This means that $S'$ contains only a single value.
  Clearly, for arbitrary rounds $r'$, $\InvB(r') \Rightarrow \InvA(r')$.
  Therefore, if $S_p'^{r+1}$ contains $\Proposal^s_q$ for
  $s \geq r+1-N \geq b$,
  because $\Locked^s_q = \true$ and $\InvA(s)$ holds,
  we have $\Proposal^s_q = v$.

  (3) $p$ does not assign a new value to $\Proposal^{r+1}_p$. Then,
  $\Proposal^{r+1} = v$, by the hypothesis that $\InvB_p(r)$ holds.
\end{proof}

With these preparations, we can now prove agreement, validity and termination
of our consensus algorithm.

\begin{lemma}
  \algref{alg:consensus} ensures agreement under each sequence $\sigma \in \BStable(D+1)$.
  \label{lem:agreement}
\end{lemma}
\begin{proof}
  We show that if a process $p$ decides $v$ in round $r$, all future decisions are
  $v$ as well.
  A decision $v$ by $p$ can only occur if $p$ executed \lineref{line:decide}, which
  means that $S_p''^r$ contains only states $q^s$ with
  $\Locked_q^s = \true$ and $\Proposal_q^s = v$.
  We show below that
  this implies that there was a sequence
  $\sigma'' \subseteq \sigma' = \left(\G^i \right)_{i=r-\nbound(D+2\nbound)}^{r} \subset \sigma$
  s.t.\ each $\G^i \in \sigma''$ has a
  $v$-locked root component and $|\sigma''| \geq D+2\nbound$.
  We can hence directly apply \lemmaref{lem:2N-v-locked}, which yields that
  $\Proposal_q^{\postr} = v$ for all $q \in \Pi, \postr \geq r$.
  This proves our claim because a decision at process $q$ in round $\postr$ can
  only be on $\Proposal_q^{\postr}$.

  We show the contrapositive of the above implication.
  Assume there is no sequence $\sigma'' \subseteq \sigma'$
  with a common $v$-locked root component and
  $|\sigma''| \geq D+2\nbound$.
  Thus, every subsequence $\sigma''' \subseteq \sigma'$ that has a
  common $v$-locked root component may have at most
  $|\sigma'''| < D+2\nbound$ and is followed immediately by a
  communication graph $\G^{r_i}$ where some process
  $q \in \Root(\G^{r_i})$
  has $\Locked^{r_i}_{q} = \false$ or $\Proposal^{r_i}_{q} \neq v$.
  Hence, we have a set $G$ of at least $n$ such communication graphs
  during $\sigma'$.
  More accurately, there is an ordered set of
  communication graphs $G = \{ \G^{r_1}, \ldots, \G^{r_n} \}$
  with $r_1 \geq r-\nbound(D+\nbound)$, $r_n \leq r$,
  and some set $X \subseteq \Pi$ of processes such that
  $X \cap \Root(\G^{r_i}) \neq \emptyset$ for every $\G^{r_i} \in G$.
  Thus $\exists q \in X \cap \Root(\G^{r_i})$ s.t.
  $\Locked^{r_i}_{q} = \false$ or $\Proposal^{r_i}_{q} \neq v$.
  By \lemmaref{lem:information-propagation}, there is hence some process
  $q \in R$ and a $\G^{r_i} \in G$ s.t.\
  $q \in \Root(\G^{r_i})$ and $q \in \CP{p}{r_1}{r_n}$
  where $r_1 \leq r_i \leq r_n$.
  Consequently, $q^{r_i} \in S_p''^r$, but then there is a state $q^{r_i}$ in $S_p''^r$
  with $\Locked_q^{r_i} = \false$ or $\Proposal_q^{r_i} \neq v$.
\end{proof}

\begin{lemma}
  \algref{alg:consensus} ensures validity.
  \label{lem:validity}
\end{lemma}

\begin{proof}
  Since decisions are only on values of $\Proposal_p^r$, which are modified
  exclusively by either assigning $\Proposal_q^s$ from different processes or
  input values $x_i$, validity follows.
\end{proof}

\begin{lemma}
  \algref{alg:consensus} eventually terminates under any sequence $\sigma$ of
  $\BStable(D+1)$.  \label{lem:termination}
\end{lemma}

\begin{proof}
  Since $\sigma \in \BStable(D+1)$, there is an earliest subsequence $\sigma' =
  \Seqr{a}{b} \subset \sigma$ 
  with $|\sigma'| = D+1$ that has a common root $R$.
  Let $v = \max_{q \in R} \Proposal_q^a$. 
  
  We prove that in any round $r \geq b$, all processes are $v$-locked
  and have $\ell^b \leq b$.
  This implies the lemma, because, by
  \lineref{line:Sprimeprime-construction},
  in round $b+N(D+2N)$,
  $S''^{b+N(D+2N)}$ contains only $v$-locked states.
  Hence, any process that has not decided yet in round $b+N(D+2N)$
  will execute \lineref{line:decide} and decide $v$.

  In order to prove this, we introduce the invariants
  $\JnvA_p(r) \Leftrightarrow \Locked^r_p = \true \land \Proposal^r_p = v
  \land (\ell^r_p = b \lor b \in \Queue_p^r) \land \ell^r_p \leq b$
  and $\JnvA(r) \Leftrightarrow \forall p \in \Pi \colon \JnvA_p(r)$,
  which state that, at the end of round $r$, $p$, resp.\ every process of $\Pi$, is $v$-locked and
  either $p$'s lock-round is $b$ or $b$ is queued at $p$.
  We use induction to show that for all $r \geq b$
  it holds that $\bigwedge_{i=b}^r \JnvA(r) \Rightarrow \JnvA(r+1)$.

  For the base case, we show that $\JnvA_p(b)$ for any $p \in \Pi$.
  According to \corollaryref{cor:interval-detection}, in round $b$, every process $p$ has
  $\Root_p^b(a) = R$ and the guard of \lineref{line:root-guard} is evaluated to true.
  We distinguish two cases:

  (1) $p$ is not $v$-locked at the beginning of its round $b$ computation, i.e.,
  $\Locked^{b-1}_p = \false$ or $\Proposal^{b-1}_p \neq v$.
  If $\Locked^{b-1}_p = \false$, clearly
  Lines~\ref{line:lock-root}, \ref{line:locked-root}, and~\ref{line:lockround}
  are executed.
  If $\Proposal^{b-1}_p \neq v$, because $\sigma'$ is the first
  subsequence with the desired properties by assumption,
  $\Root(\G^{a-1}) \neq R$.
  According to \corollaryref{cor:underapproximation},
  $R' = \Root_p^b(a-1) \neq R$, thus, by construction of $T$,
  we have $R', R \in T$ with $R' \neq R$ and hence $\abs{T} > 1$.
  Since we assumed that $\Proposal^{b-1}_p \neq v = \Proposal_p'$, also in this subcase,
  Lines~\ref{line:lock-root},~\ref{line:locked-root}, and~\ref{line:lockround}
  are executed.
  Before executing \lineref{line:guard-backoff}, we hence have $\ell_p = b$ and
  therefore none of the subsequent guards can be triggered.
  Consequently, we have that $p$ is still $v$-locked at the end of its round
  $b$ computation and $\ell^b_p = b$.

  (2) If $p$ is already $v$-locked at the end of round $b-1$,
  we have $\ell_p^{b-1} < b$ and
  $p$ executes \lineref{line:queue-append} in round $b$,
  thereby adding $b$ to $\Queue_p$.
  Since $\ell$ is not modified later, $\ell_p^b < b$.
  If $p$ does not enter the guard of \lineref{line:guard-backoff},
  $\JnvA_p(b)$ holds before $p$ executes \lineref{line:lock-hirsch}.
  Otherwise, $\Queue_p$ still contains $b$.
  Since now only states $q^s$ with $s<b$ are known to $p$,
  \lineref{line:queue-prune} cannot remove $b$ from $\Queue_p$.
  Hence, the guard of \lineref{line:queue-getLock} is true
  and $\JnvA_p(b)$ still holds before $p$ executes
  \lineref{line:lock-hirsch}.

  Since we assume that $p$ was already $v$-locked at the end of
  round $r$, this  continues to hold also in the case when 
  \lineref{line:lock-hirsch} is ever executed: Since 
  $v \in S'$, it must be the case that if $S'$ indeed
  contains a single value $x$, then $x = v$.
  Therefore,
  $\JnvA_p(b)$ holds at the end of round $b$ as asserted.

  For the induction step, we assume that $\bigwedge_{i=b}^r \JnvA(r)$
  is true for some $r \geq b$ and show that then $\JnvA_p(r+1)$ holds
  for an arbitrary process $p$.

  We first show that, if $p$ enters \lineref{line:root-guard}, it
  cannot pass the conditional of \lineref{line:guard-newRoot}.
  For this, let us distinguish two cases:

  (1) For $r \in [b,b+D]$,
  $\JnvA(r)$ implies that
  either $b=\ell_p^{r}$ or $b \in \Queue_p^{r}$ and
  $p$ is $v$-locked at the end of round $r$.
  Thus, in round $r+1$, \lineref{line:guard-newRoot} can only be
  passed if $\abs{T} > 1$. 
  Since obviously $\max(\Queue) < r$ and we assumed that $r \in [b, b+D]$,
  by the construction of $T$ in \lineref{line:T-construction},
  $T \subseteq \{ \Root_p^r(i) \mid a \leq i \leq b \}$.
  In fact,
  $T \subseteq \{ \Root(\G^i) \mid a \leq i \leq b \} \cup \{\bot\}$,
  as a consequence of \corollaryref{cor:interval-detection}.
  By the assumption that $\sigma'$ is \Rooted{R},
  after its construction in round $r$, $T= \{ R \}$
  and \lineref{line:guard-newRoot} cannot
  be entered.

  (2) For $r \geq b+D+1$, by the induction hypothesis $\JnvA(r+1-D)$ holds.
  Thus, $\Locked_p^{r} = \true$ and for any process $q$,
  $\Proposal_q^{r-D} = v$, which implies $\Proposal_p' = v$.
  Therefore, the guard of \lineref{line:guard-newRoot} cannot be passed.

  Note that, by the hypothesis $\JnvA_p(r)$,
  $\ell_p^{r+1} \leq b$,
  since \lineref{line:guard-newRoot} was not passed and
  $\ell$ is not modified after this point.
  
  Until now we have shown that $\JnvA_p(r+1)$ holds right before $p$ executes line
  \lineref{line:guard-backoff}. We again distinguish two cases:

  (1) If $\ell \geq b$, it follows from the induction hypothesis
  $\bigwedge_{i=b}^r \JnvA(i)$ that all states from round $\ell$ to round $r$
  are $v$-locked and hence
  $S$ cannot contain any state that is not $v$-locked.
  Consequently, $p$ cannot pass the guard of \lineref{line:guard-backoff}.

  (2) If, on the other hand, $\ell < b$, the hypothesis $\JnvA(r)$ implies that
  $b \in \Queue_p$. Therefore, if the guard of \lineref{line:guard-backoff} is true,
  \lineref{line:queue-getLock} is executed and \lineref{line:locked-backoff} cannot
  be entered.

  In both cases $\JnvA_p(r+1)$ still holds before $p$ reaches
  \lineref{line:lock-hirsch}.

  Observe that, according to the induction hypothesis $\JnvA(r)$, $p$
  was $v$-locked at the end of round $r$.
  Hence, $v \in S'$, and if $S'$ contains a single element $x$, it must be
  that $x=v$.
  Hence $\JnvA_p(r+1)$ also holds if $p$ executed \lineref{line:lock-hirsch}.
  Since \lineref{line:decide} cannot not invalidate $\JnvA_p(r+1)$ either, we
  are done.
\end{proof}

Lemmas~\ref{lem:agreement},~\ref{lem:validity}, and~\ref{lem:termination} yield
the correctness of \algref{alg:consensus}:

\begin{theorem}
  \algref{alg:consensus} correctly solves consensus under $\BStable(D+1)$,
provided $N \geq n$.
  \label{thm:correctness}
\end{theorem}


\newcommand{\NONSPLIT}{\MASty{NON-SPLIT}}
\newcommand{\UNIFORM}{\lozenge \MASty{UNIFORM}}
\newcommand{\STAR}{\lozenge \MASty{STAR}}


\section{Solving Consensus for Long Periods of Eventual Stability}
\label{sec:completepicture}

While the ability to solve consensus under short-lived periods of stability 
is beneficial in terms of assumption coverage and fast termination, it
comes at the price of the rather involved \algref{alg:consensus}. In this
section, we show that a (considerably) longer period of stability facilitates 
a (considerably) simpler algorithm. More specifically, 
we show that, for $\EStable(x)$ with $x \geq 3n-3$, 
a simple consensus algorithm can be obtained by adopting
existing algorithmic techniques \cite{CBS09}.
Unfortunately, though, this approach does not work for $x < 3n-3$.

With $\Sigma$ representing all possible communication graph sequences,
let $\NONSPLIT$ be defined as those $\sigma \in \Sigma$ for which the
following holds:
If $\Gr \in \sigma$, then for all $p, q \in \Gr$, there is some $q' \in \Gr$ such
that $(q' \ra p) \in \Gr$ and $(q' \ra q) \in \Gr$.
Furthermore, let $\UNIFORM$ represent those sequences where there is a round
$r$ and a set $P$ of processes s.t.\ $\Gr = \bigcup_{p \in P} \Starwl{p}$, with
$\Starwl{p}$ denoting the star graph with central node $p$ and no self-loops in
the non-center nodes.
Finally, let $\STAR(y)$ be the set of those $\sigma \in \Sigma$ that contain a
subsequence $\sigma' \subseteq \sigma$, $|\sigma'| \geq y$, with a \VSRC{}
$R$ satisfying $\bigcup_{p \in R} \Starwl{p} \subseteq \Gr$ for each $\Gr \in \sigma'$.

Note that the definition of $\UNIFORM$ is such that if $(q \ra q) \in \Gr$
then $q \in P$, i.e., processes outside of $P$ cannot have any outgoing edges
(including self-loops).
This restriction does not hold for $\STAR(y)$, however, which is hence a
strictly stronger adversary for $y=1$, i.e., $\STAR(1) \supset \UNIFORM$.

Given a sequence $\sigma =  (\G^i)_{i>0}$, let the compound sequence
$\tilde{\sigma} = (\tilde{\G}^i)_{i > 0}$ be the sequence of graphs
$\tilde{\Gr} = \G^{(r-1)(n-1)+1} \circ \cdots \circ \G^{r(n-1)}$.
We recall from \cite[Lemma 4]{CFN15:ICALP} that if $\sigma \in \BSafety$ then
$\tilde{\sigma} \in \NONSPLIT$.

Since the compound graphs $\tilde{\Gr}$ combine $n$ consecutive graphs
$\Gr$ in fixed (ascending) order, this construction is not necessarily aligned
with the $R$-rooted subsequence provided by $\BStable(y)$. 
Let $\tilde{\sigma}$ be the compound sequence of some $\sigma \in \BStable(y)$.
In the worst case, $\tilde{\sigma}$ may contain a sequence of only $\lfloor
\frac{y-n+1}{n-1} \rfloor$ compound graphs $\tilde{\Gr}$ with the property
$\bigcup_{p \in R} \Starwl{p}\subseteq \tilde{\Gr}$, i.e., $\tilde{\sigma} \in
\STAR(\lfloor \frac{y-n+1}{n-1} \rfloor)$.

From \cite[Theorem 5]{CBS09}, we know that consensus can be solved under
$\NONSPLIT \cap \UNIFORM$.
On the other hand, our impossibility result \theoremref{thm:D-imposs} (for $D=n-1$),
implies that consensus is impossible under the adversary $\NONSPLIT \cap \STAR(1)$.
Note carefully that we can allow the $R$-rooted subsequence to be aligned
with the compound graph construction here, hence the impossibility
for $\STAR(1)$.

Consequently, for $i < 3$, using compound graphs does not provide
a solution for $\BSafety \cap \BLiveness(i(n-1))$, as
$\sigma \in \BSafety \cap \BLiveness(i(n-1))$ implies only $\tilde{\sigma}
\in \STAR(i-1) \cap \NONSPLIT$ according to the considerations above.
For $i \geq 3$, however, we have that $\tilde{\sigma} \in \STAR(2) \cap \NONSPLIT$,
which allows us to derive a simple consensus algorithm 
\algref{alg:voting} by combining the detection of root components
with the uniform voting algorithm from \cite[Algorithm 5]{CBS09}.

\begin{algorithm}[ht!]
  \small
  \caption{Compound graph algorithm, code for process $p$}
  \label{alg:voting}
  \DontPrintSemicolon
  \SetKwInput{Initialization}{Initialization}
  \SetKwInput{Transmit}{Transmit round $r$ messages}
  \SetKwInput{Compute}{Perform round $r$ computation until decision}

  \Initialization{}
  $\Proposal \gets x_p$ \;
  $m \gets \bot$ \;
  \BlankLine

  \Transmit{}
  Send $(m, \Proposal)$ to all $q$ with $(p \ra q) \in \Gr$ \;
  Receive $(m_q, \Proposal_q)$ from all $q$ with $(q \ra p) \in \Gr$ \;

  \BlankLine

  \Compute{}
  
  $M \gets \bigcup_{(q \ra p) \in \Gr} (q, m_q, \Proposal_q)$ \;
  $R^{r-1} \gets \Root^r_p({r-1})$ \;

  \uIf{$\bigcup_{(q, m_q, \Proposal_q) \in M} m_q  = \{  v \}$ with $v \neq \bot$}{
    $\Proposal \gets v$ \;
    $m \gets v$ \;
    decide $v$ \;
  }
  \uElseIf{$R^{r-1} \neq \emptyset$}{
    \uIf{$q \in R^{r-1}$ with $(q, m_q, \Proposal_q) \in M$ s.t.\ $m_q \neq \bot$}
    {
      $\Proposal \gets m_q$ \;
    }
    \uElse{
      $\Proposal \gets \max \{ \Proposal_q \mid (q, m_q, \Proposal_q) \in M, q \in R^{r-1} \}$ \;
    }
    $m \gets \Proposal$ \;
  }
  \uElseIf{$(q, m_q, \Proposal_q) \in M$ with $m_q \neq \bot$}{
    $\Proposal \gets m_q$ \;
    $m \gets \bot$ \;
  }
  \uElse{
    $m \gets \bot$ \;
  }
\end{algorithm}

\begin{theorem}
  \algref{alg:voting} solves Consensus under the adversary $\STAR(2) \cap \NONSPLIT$.
  \label{thm:voting}
\end{theorem}

\begin{proof}
If process $q$ is the first to decide $v$, at the end of round $r$,
then every process has $\Proposal^r = v$:
Since $q$ received only messages $(v, *)$ with $v \neq \bot$ and
$\Gr$ is non-split, any process $p$ received a message $(v, *)$.
Moreover no message with $(v',*)$ and $v' \neq v$ was sent in round $r$,
because every message with some $m\neq \bot$ was generated based on the root
of round $r-1$ which is, if detected by a process, detected consistently
on every process.
Hence during the round $r$ computation, every process sets $\Proposal$ to
$v$ either via line $13$ or line $18$.
As all future decisions are based on $v$, this ensures agreement.

%

Termination occurs at the latest when the two successive graphs $G, G'$ of $\STAR(2)$
occur in rounds $r'$, $r'+1$: In round
$r'+1$, every process detects $R$, the root component of $G$ and sets
$\Proposal = v$, where $v$ is uniquely determined by $R$.
Subsequently, any process sends and receives only messages $(*, v)$.
\end{proof}

\section{Conclusions}\label{sec:conclusion}

We provided tight upper and lower bounds for the solvability of consensus
under message adversaries that guarantee a stable root component only
eventually and only for a short period of time:
We showed that consensus is solvable if and only if each graph has exactly one
root component and, eventually, there is a period of at least $D+1$
consecutive rounds (with $D\leq n-1$ denoting the number of rounds required by
root members for broadcasting) where the root component remains the same.
We also provided a matching consensus algorithm, along with its correctness
proof.
While this kind of short-lived periods of stability is useful from the perspective of
assumption coverage in real systems and fast termination time, we also demonstrated
that longer periods of stability allow the development of less complex algorithms.

\clearpage
\bibliography{lit_bib/lit}
\bibliographystyle{abbrv}

\iftoggle{TR}{}{
\clearpage
\input{appendix}}

\end{document}